\documentclass{amsart}[12pt]
\usepackage{amsmath,amscd}
\usepackage[utf8]{inputenc}
\usepackage{multicol}
\usepackage{url}
\usepackage{xcolor}
\usepackage{colortbl}

\newtheorem{example}{Example}
\newtheorem{conjecture}{Conjecture}
\newtheorem{remark}{Remark}

\newtheorem{proposition}{Proposition}
\newtheorem{fact}{Numerical fact}
\newtheorem{problem}{Problem}

\def\B{\mathcal B}

\def\rmc(#1,#2){RM(#1,#2)}
\def\form(#1,#2){H(#1,#2)}
\def\val(#1,#2){V(#1,#2)}
\def\boole(#1){ B(#1) }
\def\bst(#1,#2,#3){\boole(#1,#2,#3)}
\def\tildeboole(#1){ \widetilde B(#1) }

\def\fd{{\mathbb F}_2}
\def\fdm{{\mathbb F}^m_2}

\def\aglm{{\textsc{agl}}(m,2)}

\def\der(#1,#2){{\rm Der}_{#2} (#1)}

\def\stab(#1){\textsc{stab}(#1)}
\def\stablevel(#1,#2){\textsc{stab}_{#1}(#2)}

\def\rmq(#1,#2,#3){\rmc(#1,#2)/\rmc(#3,#2)}

\def\walsh(#1,#2){\widehat{#1}(#2)}

\def\level(#1){\underset{#1}{=}}

\def\binogauss(#1,#2){\genfrac{[}{]}{0pt}{}{#1}{#2}}
\newcommand{\card}[1]{\vert{#1}\vert}
\def\val(#1){{\rm val}(#1)}

\def\anf(#1){{\rm anf}(#1)}
\def\fix(#1,#2,#3,#4){{{\rm fix}^{#1,#2}_{#3}}{(#4)}}
\def\classe(#1,#2,#3){{{\mathcal T}(#1,#2,#3)}}

\def\stab(#1){\textsc{stab}(#1)}
\def\stablevel(#1,#2){\textsc{stab}^{#1}(#2)}
\def\stableveldim(#1,#2,#3){\textsc{stab}^{#1}_{#2}(#3)}
\def\stableveldeg(#1,#2,#3,#4){\textsc{stab}^{#1,#2}_{#3}(#4)}
\def\level(#1){\underset{#1}{\sim}}
\def\bound(#1,#2){\underset{#1}{\overset{#2}{\sim}}}
\def\modulo(#1,#2){\mod\rmc(#1,#2)}

\def\pow#1.#2{\tiny$10^{#1.#2}$}

\def\todo#1{\begin{center}TODO : #1\end{center}}

\def\level(#1){\underset{#1}{=}}
\def\val(#1){{\rm val}(#1)}

\def\anf(#1){{\rm anf}(#1)}
\def\fix(#1,#2,#3,#4){{{\rm fix}^{#1,#2}_{#3}}{(#4)}}
\def\classe(#1,#2,#3){{{\mathcal T}(#1,#2,#3)}}

\def\stab(#1){\textsc{stab}(#1)}
\def\stablevel(#1,#2){\textsc{stab}^{#1}(#2)}
\def\stableveldim(#1,#2,#3){\textsc{stab}^{#1}_{#2}(#3)}
\def\stableveldeg(#1,#2,#3,#4){\textsc{stab}^{#1,#2}_{#3}(#4)}
\def\level(#1){\underset{#1}{\sim}}
\def\bound(#1,#2){\underset{#1}{\overset{#2}{\sim}}}
\def\modulo(#1,#2){\mod\rmc(#1,#2)}

\def\pow#1.#2{\tiny$10^{#1.#2}$}

\def\degV(#1,#2,#3){\deg_{#1+#2}(#3)}

\title{On the Normality of Boolean quartics}

\author{Valérie Gillot}
\author{Philippe Langevin}
\address{Imath, university of Toulon}
\email{\{valerie.gillot,philippe.langevin\}@univ-tln.fr}
\author{Alexandr Polujan}
\address{Otto-von-Guericke-Universit\"{a}t}
\email{alexandr.polujan@gmail.com}

\date{April 30th 2024}

\begin{document}

\maketitle
\begin{abstract}
In the BFA'2023 conference {\color{black} paper}, A. Polujan, L. Mariot and S. Picek {\color{black} exhibited the first} example of a non-normal {\color{black} but} weakly normal bent function in dimension 8. In this note, we present numerical approaches  based on the classification  of Boolean spaces to explore in detail the normality of bent functions of 8 variables and we complete S.~Dubuc's results for dimensions less or equal to 7. Based on our investigations, we show that all bent functions in 8 variables are normal or weakly normal. Finally, we conjecture that more generally all Boolean functions of degree at most 4 in 8 variables are normal or weakly normal. 
\end{abstract}

\section{Introduction}

In the context of random generation of non-normal bent {\color{black} functions in 8 variables, a non-normal but weakly normal bent quartic} was presented by A. Polujan, L. Mariot and S.  Picek at {\color{black}BFA'2023; see~\cite{PMP_BFA23}}. Furthermore, a non-normal function of degree 6 in 8 variables {\color{black} was} presented by S. Dubuc in her PhD thesis, and we verify that this function is in fact weakly normal. For that dimension, the main question concerns the existence
of Boolean functions that are neither normal nor weakly normal, we will say abnormal.  For that purpose, we propose 
a numerical approach based on the classification of Boolean spaces under the action of the affine general linear group
to study the relative degree of Boolean functions. We namely use most of the old  \cite{PLGL-2011} 
and recent \cite{bfapaper2022} classification results obtained in the last decade to analyse the 
maximal relative degree that a Boolean function can have. Our main objective is to update the knowledge
of Boolean functions in small dimensions, thus completing the {\color{black}results in the} PhD work of S. Dubuc \cite{dubuc},  
and answering definitively on the (weak)-normality of bent function and related questions 
in 8 variables.

Let $\fd$ be the finite field of order $2$. Let $m$ be a positive integer. 
We denote $\boole(m)$ the set of Boolean functions $f \colon \fdm \rightarrow\fd$. Every Boolean 
function has a unique algebraic reduced representation:
\begin{equation}\label{ANF}
f(x_1, x_2, \ldots, x_m ) = f(x) = \sum_{S\subseteq \{1,2,\ldots, m\}} a_S X_S,
\quad a_S\in\fd, \ {\color{black} X_S = \prod_{s\in S} x_s}.
\end{equation}
The degree of $f$ is the maximal cardinality of $S$ with  $a_S=1$ in
the algebraic form. In this paper, we conventionally fix the degree
of the null function to zero. We denote by $\boole(s, t, m)$ the space of Boolean
functions of valuation greater than or equal to $s$ and of degree less 
than or equal to $t$. Note that  $\boole(s,t,m)=\{0\}$ whenever $s>t$.
The affine general  linear group of $\fdm$, denoted by $\aglm$, acts naturally 
over all these spaces. A system of representatives of $\boole(s,t,m)$ is denoted $\tildeboole(s,t,m)$. In next sections, we are using the classifications 
available in the project \cite{project}, more precisely those which interest us here are listed in Table \ref{CLASS}.

 \begin{table}[h]
 \small
	\caption{\label{CLASS} Useful spaces and the corresponding number of classes}
\begin{tabular}{|c|c|c|c|c|c|c|c|}
 \hline 
  $\boole(s,t,m)$&$\bst(1,5,5)$& $\bst(1,6,6)$
   &$\bst(1,3,7)$  
     &$\bst(4,7,7)$
     &$\bst(2,3,8)$
     &$\bst(4,4,8)$\\
     \hline
 $\sharp \tildeboole(s,t,m)$ & 206 &$7\,888\,299$
   &$1\,890$ 
     &$68\,443$ 
         &$20\,748$
          &$999$\\
 \hline
 \end{tabular}
 \end{table}

\section{Normality and Relative degree}

Let $V$ be a {\color{black} linear} subspace of $\fdm$, we denote $\B(V)$ the space of Boolean functions from $V$ to $\fd$. For $a\in \fdm$, we denote $f_{a,V}$ the Boolean function of $\B(V)$ defined by $v\in V \mapsto f(a+v)$.

{\color{black}The notion of normality was introduced by Dobbertin~\cite{DobbertinFSE1994}. In this paper,} we use the following terminologies of P. Charpin \cite{CHARPIN}.
A Boolean function $f\in\boole(m)$ is said to be \textsl{normal} if there exists an affine subspace $V+a$ of $\fdm$ with dimension $\lceil m/2 \rceil$ where $f_{a,V}$ is constant i.e. $\deg(f_{a,V})= 0$. 
The function $f$ is \textsl{weakly normal} when $f_{a,V}$ is affine, and not constant i.e. $\deg(f_{a,V})=1$, on an affine space $V+a$ of dimension  $\lceil m/2 \rceil$. 
\begin{example}
For a positive integer $t$, the quadric $x_1x_{t+1}+ x_2x_{t+2}+ \cdots +x_tx_{2t}+x_{2t+1}$ 
is not normal, but weakly normal in $\boole(2t+1)$.
\end{example}

The normality of Boolean functions is an invariant under the action of  $\aglm$.

Now, we introduce the relative degree that generalizes normality notions.
The \textsl{relative degree} of $f$ with respect the affine space $a+V$ is the degree of $f_{a,V}$,  
denoted by: $\degV(a,V,f)=\deg(f_{a,V})$. Note that if $f_{a,V}=0$, then $\degV(a,V,f)=0$.
We define the \textsl{$r$-degree} of $f\in\boole(m)$ as the minimal relative degree of $f$ 
for all affine spaces $a+V$, where $\dim (V)=r$:
$$0 \leq \deg_r(f)=\min\{\degV (a,V,f) \mid \dim V=r\quad \text{and}\quad a\in \fdm \}.$$
Notions of normality are translated into
$$\deg_{\lceil m/2 \rceil}(f)=\begin{cases}
0, & \text{$f$ is normal};\\
1,& \text{$f$ is non-normal and weakly normal};\\
\geq 2,& \text{ $f$ is abnormal}.\\
\end{cases}
$$
For integers $r\leq m$ and $k\leq m$, we {\color{black} propose} to study the following combinatorial parameter
$$D_r( k, m) = \max \{\deg_r( f ) \mid  \deg(f) \leq  k \}.
$$

As we will see,  it may be interesting to look at more closely what is happening by degree, for that the denote by $D_r^\dag(k, m)$ the maximum  over the Boolean functions of degree equal to $k$.
In general, the determination of $D_r(k, m)$ seems to be  a hard problem. Since
the affine general linear group acts on the set of affine space of 
given dimension preserving the degree of functions, whence the relative degrees are affine invariants of $\boole(m)$. So,
we can use the result of classification of Boolean functions to determine some of these parameters. Each space of dimension $r$ ($r$-space) of $\fdm$ has $2^{m-r}$ cosets and the number of $r$-spaces is
given by  Gauss binomial coefficient  $\binogauss(m,r)$. Taking into account 
the cost of the computation of the degree of a Boolean function, 
the work factor to determine the $r$-relative degree of function in $\boole(s,t,m)$
by brute force:
$$
        W( r, s, t ,m) = \sharp \tildeboole(s,t,m) \times  2^{m-r} \binogauss(m,r)\times r2^{r},\quad   
        \binogauss(m,r)=\prod_{i=0}^{r-1} \frac{2^m - 2^{i}}{2^r - 2^i}.
$$
For $m=6$, $r=4$ and $\tildeboole(1,6,6)=7\,888\,299 \approx 2^{23}$, the work factor is about $2^{40}$. So, we complete the tables $D^\dag_r( k, 6)$ and $D^\dag_r( k, 5)$ in a matter of minutes, using a 80-core computer. But,
since  $\tildeboole(2,4,7) = 118\,140\,881\,980 \approx 2^{37}$, 
we can not reasonably use brute force in dimension $7$.

\begin{table}[h]
	\small
	\caption{\label{D-values} The values of $D_r^\dag(k, m)$, for a fixed $m$}
\setlength{\columnsep}{0cm}
\begin{multicols}{2}

\small

 \begin{center}
 {\label{DR5}$D^\dag_r( k, 5)$}

\begin{tabular}{|c|ccccc|}
	\hline
$r\backslash k$  &1 &2 &3 &4 &5\\
 \hline
 \hline
 4 &0&2 &2 &3 & 2\\
 \rowcolor[gray]{.8}
 3 &0 &1 &1 &1 & 1\\

2 &0 &0 &0 &0 & 0\\
 \hline
 \end{tabular}
 \end{center}
\columnbreak
 \begin{center}
{\label{DR6}$D^\dag_r( k, 6)$}

\begin{tabular}{|c|cccccc|}
	\hline
$r\backslash k$  &1 &2 &3 &4 &5 & 6 \\
 \hline
 \hline
 5 & 0 & 2&	3& 	4& 	3&	4 \\
 4 &0 &	2&	2&	2& 	2&	2\\
 \rowcolor[gray]{.8}
 3 &0&	0&	0 &	0& 	0&	0\\
 2 &0&	0&	0&	0 &	0&	0\\
 \hline
 \end{tabular}
 \end{center}
\end{multicols}
 \begin{center}
{\label{DR7}$D^\dag_r( k, 7)$}

\def\rand#1{\textcolor{red}{$\geq #1$}}
\begin{tabular}{|c|ccccccc|}
	\hline
 $r\backslash k$ 	&1 	&2 	&3 	&4 	&5 	&6	&7\\
 \hline
 \hline
6	&0	&2 	&3	&\rand 4 	&\rand 4 	&{5}	    &\rand 4\\
5	&0	&2	&3	&\rand 3 &\rand 3 	&\rand 3	&\rand 3\\
\rowcolor[gray]{.8}
4	&0	&1	&1	&\rand 1 &\rand 1	&\rand 1	&\rand 1\\
3	&0	&0	&0	&0	&0 	&0	&0\\
2	&0	&0	&0	&0  &0 	&0	&0\\
\hline
  \end{tabular}
 \end{center}
\end{table}

\begin{fact}[]
\label{NORMAL6}
Using the classification of $\boole(1,6,6)$, we complete the 
table $D^\dag_r( k, 6)$. 
As proved in \cite{dubuc}, all the functions  in $\boole(6)$ are normal. We point
this implies that  $D_3(k,m)=0$, for all $m\geq 6$. 
\end{fact}

Among the 12 members of $\tildeboole(5,7,7)$, all have 6-relative degree 0 except
$$h(x)=x_1 x_2 x_3 x_4 x_5 x_6 + x_2 x_3 x_4 x_5 x_7 + x_1 x_3 x_4 x_6 x_7 + x_1 x_2 x_5 x_6 x_7,$$
for which $\deg_6(h) = 5$ fixing the entry  (6,6) of the table $D^\dag_r( k, 7)$.

\begin{fact}\label{CUBICFACT7}
Using the classification of $\boole(1,3,7)$, we fill the left part of the table $D^\dag_r( k, 7)$.  
All {\color{black} cubics} in  $\boole(7)$ are normal or weakly normal.
\end{fact}

The minoration in the right  part of the table $D^\dag_r( k, 7)$ were obtained at random,
by adding a random cubic to the members of $\tildeboole(4,7,7)$. The computation
suggests the non-existence of abnormal quartic~:
\begin{conjecture}
All the quartics of $\boole(7)$ are normal or weakly normal.
\end{conjecture}

\begin{remark}
Verifying the previous conjecture, which consists of listing all the functions of the form form $q+h$ with $q \in \boole(2,2,7)$ and $h \in \tildeboole(3,4,7)$ i.e. $2^{21} \times 2^{16}$ functions, does not seem numerically out of reach.
\end{remark}

\section{Normality in dimension 8}

In her PhD thesis, S. Dubuc gave and presented the first example of a non-normal function. It has degree 6 comprising 140 monomials, but we checked that this function is weakly normal and equivalent {\color{black} to the following one:}

\begin{equation*}
	\begin{split}
		f(x)=&	x_2 x_3 x_4 x_5 x_7 x_8   + x_2 x_3 x_4 x_5 x_8 + x_2 x_3 x_4 x_7 x_8 + x_2 x_3 x_4 x_6 + x_2 x_3 x_5 x_6 \\
		+& x_2 x_3 x_4 x_8 + x_2 x_4 x_6 x_8 + x_2 x_5 x_7 x_8 + x_1 x_2 x_3 + x_3 x_4 x_5+ x_2 x_5 x_6  + x_2 x_4 x_7 \\ 
		+& x_3 x_4 x_7 + x_4 x_5 x_8 + x_3 x_6 x_8 + x_3 x_7 x_8 + x_2 x_3 + x_2 x_5 + x_3 x_5 + x_4 x_6 + x_2 x_7 \\ 
		+& x_3 x_8 + x_7 x_8 + x_1 + x_2.
	\end{split}
\end{equation*}


She also proved that all the bent {\color{black}cubics} in 8 variables are normal. A fact that we can improve strongly in two ways, by proving that all 8-bit cubics are normal or weakly normal and by proving that all 8-bit bent functions are normal or weakly normal.

\begin{problem} 
    Does there exist a Boolean function $\boole(8)$ with 4-relative degree 2? 
\end{problem}

\begin{proposition} 
All the {\color{black}cubics} in $\boole(8)$ are normal or weakly normal.
\end{proposition}
\begin{proof}
The distribution by relative degrees of the $20\,748$ representatives of $\boole(2,3,8)$ is given in  Table \ref{CUBIC}. For example, the entry (6,1) which is 21 means that there are 21 cubics of 6-relative degrees equal to 1. The distribution for $r=4$ gives the result.
\end{proof}

\begin{center}
	\begin{table}[h]
		\caption{\label{CUBIC} Distribution of relative degrees of cubics in $\tildeboole(2,3,8)$ }
		\begin{tabular}{|c|c|c|c|c|c|}
			\hline
			$r\backslash \deg_r$ & 0 & 1 & 2 & 3\\
			\hline
			8   &$1$ &$0$ &$0$ &$20\,747$\\
			7   &$10$   &$0$    &$53$    &$20\,712$\\
			6   &$130$  &$21$  &$1\,910$  &$18\,714$\\
			5   &$5\,504$ &$5\,227$ &$10\,044$ &$0$ \\
			4   &$20\,748$ &$0$   &$0$     &$0$\\
			\hline
		\end{tabular}
	\end{table}
\end{center}

\begin{remark}
An alternative proof of the previous proposition is to use the numerical fact \ref{CUBICFACT7}. Indeed, for $f \in \rmc(3,8)$ and for any hyperplane $H\subset \fd^8$, the restriction of $f$ to $H$ is a cubic in $\boole(H)\sim \boole(7)$.
\end{remark}

We say {\color{black}that a} homogeneous quartic $h$ {\color{black} in $B(m)$} is \textsl{bentable}, if there exists a Boolean cubic $c$ {\color{black} in $B(m)$} such that $h+c$ is bent.  In that case, we say that $h+c$ has type $S_N$, where $N$ is the order of the stabilizer  of $h$ in $\boole(4,4,8)$.  It is known \cite{PLGL-2011}  that 536 of the 999 classes of $\boole(4,4,8)$, are bentable and they can be arranged in 418 classes {\color{black} using the notion of the dual bent function}. At  BFA'2023,  A. Polujan, L. Mariot and S.  Picek~\cite{PMP_BFA23} presented the first example of a non-normal bent quartic in 8 variables. This example belongs to the partial spread class and has type $S_1$.

\begin{fact}
It is known that up duality, there are 25 bentable homogeneous quartics of type $S_1$, providing   $26\,532\,848 \approx 2^{24.7}$ classes of bent functions, in which we {\color{black} identify 361 classes of non-normal but weakly normal} bent functions.
\end{fact}

To decide on the normality of a large set of 8-bit bent functions, we
use a rather approach initialized by the pre-computation of the
$3\,108\,960$ affine 3-spaces as a list of $97\,155$  3-spaces and their 32 cosets. 
A Boolean function is weakly normal is there exists 
a pair of cosets of the same 3-space on which $f$ is constant. If $f$ takes
the same value on these cosets then it is normal. Using a bit programming
implementation in C, we were able to show the following result.

\begin{fact} 
\label{bentB8}
All the bent {\color{black} functions} in $\boole(8)$ are normal or weakly normal. 
\end{fact}

We have adapted the counting method presented in \cite{PLGL-2011}  to enumerate a cover-set of 8-bit bent functions of degree 4. 
This cover set  contains  $355\,073\,617\approx 2^{28.52}$  bent functions  and data are available in \cite{bentproject}. 
Checking the normality or weak normality of these functions, we were able to otain fact \ref{bentB8}. Since 8-bit bent 
functions have a degree of at most 4, the previous result leads us to consider a more general phenomenon~: 

\begin{conjecture} 
All the Boolean {\color{black} quartics} in $\boole(8)$ are normal or weakly normal. 
\end{conjecture}

In the space of homogeneous quadratic forms of 8 variables, 
a space of dimension 28, the degree of Boolean indicator 
of bent functions has degree 4 in \boole(8). This last conjecture is an attempt
to precise the structure of the set of bent functions 
within the space of components of a 8-bit APN function.

\section*{Acknowledgments}
Philippe Langevin is partially supported by the French Agence Nationale 
de la Recherche through the SWAP project under Contract ANR-21-CE39-0012.

\nocite{*}
\bibliographystyle{plain} 
\bibliography{normal.bib}

\begin{thebibliography}{1}

\bibitem{BENT}
Anne Canteaut, Magnus Daum, Hans Dobbertin, and Gregor Leander.
\newblock Finding nonnormal bent functions.
\newblock {\em Discret. Appl. Math.}, 154(2):202--218, 2006.

\bibitem{CHARPIN}
Pascale Charpin.
\newblock Normal {B}oolean functions.
\newblock {\em Journal of Complexity}, 20(2):245--265, 2004.
\newblock Festschrift for Harald Niederreiter, Special Issue on Coding and
  Cryptography.

\bibitem{DobbertinFSE1994}
Hans Dobbertin.
\newblock Construction of bent functions and balanced {B}oolean functions with
  high nonlinearity.
\newblock In Bart Preneel, editor, {\em Fast Software Encryption}, pages
  61--74, Berlin, Heidelberg, 1995. Springer Berlin Heidelberg.

\bibitem{dubuc}
Sylvie Dubuc.
\newblock {\em Etude des propriétés de dégénérescence et de normalité des
  fonctions booléennes et construction de fonctions q-aires parfaitement
  non-linéaires}.
\newblock PhD thesis, University of Caen, 2001.

\bibitem{bfapaper2022}
Val\'erie Gillot and Philippe Langevin.
\newblock Classification of some cosets of {R}eed-{M}uller codes.
\newblock {\em Cryptography and Communications}, 15:1129--1137, 2023.
\newblock doi:10.1007/s12095-023-00652-4.

\bibitem{bentproject}
Valérie Gillot and Philippe Langevin.
\newblock Classification of 8-bit bent functions.
\newblock \url{https://langevin.univ-tln.fr/project/genbent/genbent.html}.

\bibitem{project}
Valérie Gillot and Philippe Langevin.
\newblock Classification of $\boole(s,t,m)$.
\newblock \url{http://langevin.univ-tln.fr/data/bst/}.

\bibitem{PLGL-2011}
Philippe Langevin and Gregor Leander.
\newblock Counting all bent functions in dimension eight
  99270589265934370305785861242880.
\newblock {\em Designs, Codes Cryptography 59(1-3)}, 59(1-3):193--205, 2011.

\bibitem{PMP_BFA23}
Alexandr Polujan, Luca Mariot, and Stjepan Picek.
\newblock Normality of {B}oolean bent functions in eight variables, revisited.
\newblock In {\em The 8th International Workshop on {B}oolean Functions and
  their Applications}, pages 79--83, 2023.

\end{thebibliography}

\end{document}